\documentclass[11pt]{article}
\usepackage{xcolor}
\definecolor{darkgreen}{rgb}{0,0.5,0}
\usepackage[colorlinks, citecolor=blue, linkcolor=black]{hyperref}
\usepackage{fullpage,amssymb,amsmath,amsthm}
\usepackage[linesnumbered,boxed,noend]{algorithm2e}

\newtheorem{THEOREM}{Theorem}
\newtheorem{COROLARY}[THEOREM]{Corollary}
\newtheorem*{nonumtheorem}{Theorem}
\newtheorem{theorem}{Theorem}[section]

\newtheorem{conjecture}[theorem]{Conjecture}

\newtheorem{definition}[theorem]{Definition}

\newtheorem{lemma}[theorem]{Lemma}

\newtheorem{remark}[theorem]{Remark}

\newcommand{\set}[2][]{{\left\{#2\right\}}^{#1}}
\newcommand{\bools}[1]{\set[#1]{0,1}}

\newcommand{\abs}[1]{\left|{#1}\right|}
\newcommand{\size}[1]{\abs{#1}}

\newcommand{\half}{\frac{1}{2}}

\newcommand{\N}{\mathbb{N}}

\newcommand{\Image}[1]{\mathrm{Im} \left( #1 \right) }

\newcommand{\poly}{\mathrm{poly}}

\newcommand{\ignore}[1]{}

\newcommand{\BPP}{\mathsf{BPP}}
\newcommand{\PTIME}{\mathsf{P}}

\newcommand{\ppoly}{\mathsf{P/poly}}

\newcommand{\AM}{\mathsf{AM}}
\newcommand{\SZK}{\mathsf{SZK}}
\newcommand{\ZPP}{\mathsf{ZPP}}
\newcommand{\NP}{\mathsf{NP}}
\newcommand{\coAM}{\mathsf{coAM}}
\newcommand{\pr}{\mathsf{pr}}

\newcommand{\MCSP}{\mathsf{MCSP}}
\newcommand{\MTKP}{\mathsf{MTKP}}
\newcommand{\Avoid}{\mathsf{Avoid}}
\newcommand{\HSP}{\mathsf{HSP}}

\begin{document}

\title{A Note on $\Avoid$ vs $\MCSP$}

\author{%
Edward A. Hirsch\thanks{Department of Computer Science, Ariel University, Israel. This research was conducted with the support of the State of Israel, the Ministry of Immigrant Absorption, and the Center for the Absorption of Scientists. Email: {\tt edwardh@ariel.ac.il}} \and Ilya Volkovich\thanks{Computer Science Department, Boston College, Chestnut Hill, MA. 
    Email: {\tt ilya.volkovich@bc.edu}}%
}

\date{}

\maketitle
\newcommand\EH[1]{\textcolor{darkgreen}{[EH: #1]}}

\begin{abstract}
A recent result of Ghentiyala, Li, and Stephens-Davidowitz (ECCC TR 25-210) shows that any language reducible to the Range Avoidance Problem ($\Avoid$) via deterministic or randomized Turing reductions is contained in $\AM \mathrel{\cap} \coAM$. In this note, we present a different potential avenue for obtaining the same result via the Minimal Circuit Size Problem ($\MCSP$).
\end{abstract}

\section{Introduction}

The \emph{Range Avoidance} ($\Avoid$) problem is defined as follows: given a Boolean circuit $C: \bools{n} \to \bools{m}$ with $m > n$, output a string $y \in \bools{m}$ such that $y \not \in \Image{C}$, i.e., such that $y$ does not have a pre-image w.r.t. $C$. The problem has been studied since 1980s in mathematical logic literature \cite{PWW88corr}, where it is referred to as the ``dual Weak Pigeonhole Principle'' (dWPHP). Its ``new life'' in the computational complexity theory has begun with the work of  Kleinberg, Korten, Mitropolsky, and Papadimitriou \cite{KKMP21}. Since then, the problem attracted significant attention 
(see e.g. \cite{GLS25} and the references therein for a survey from the complexity-theoretic perspective, and \cite{KraGen} for the body of work in the field of bounded arithmetic).

$\Avoid$ is a natural example of a total search problem, that is, a search problem for which a solution always exists. Furthermore, a randomly selected element of $\bools{m}$ is a solution with probability at least $1 - 2^{-(m-n)}$. Thus, the hardest instances\footnote{There is a stronger version corresponding to the dual version of the classical pigeonhole problem where the domain is just one element smaller than the image. This version has also been studied in the literature, yet its complexity properties are quite different.} occur when $m = n+1$, in which case the trivial randomized algorithm could succeed with probability close to $\half$. In order to amplify the success probability, one could, naturally, draw several random samples and check them. This will result in a ``zero-error'' $\ZPP^{\NP}$ algorithm. The $\NP$ oracle, however, appears to be necessary, since it is unclear how to \emph{efficiently} verify or select a correct solution, even when one is presented. This raises the following natural question:
\begin{quote}
\emph{``Can we amplify the success probability of solving $\Avoid$ without an $\NP$ oracle, or at least using an oracle for an easier problem''?}
\end{quote}

In the \emph{Minimum Circuit Size Problem} ($\MCSP$), we are given the truth table $tt \in \bools{2^n}$ of an $n$-variate Boolean function $f : \bools{n} \to \bools{}$ and a parameter $0 \leq s \leq 2^n$, and the goal is to determine whether $f$ can be computed by a Boolean circuit of size at most $s$. The problem shares a similar fate to $\Avoid$: although its origins trace back to the 1960s (see e.g. \cite{Trakhtenbrot84corr}), the interest in the problem was renewed following its reintroduction by Kabanets and Cai in \cite{KabanetsCai00}.
In that work, it was also observed that $\MCSP \in \NP$. However, the true complexity of the problem remains unknown. In particular, it is neither known nor believed to be $\NP$-hard under standard (deterministic) many-to-one reductions \cite{MW15corr,HP15}. Nevertheless, several recent results provide evidence of $\NP$-hardness of $\MCSP$ under more general notions of reduction \cite{ILO20,Ilango20b,Ilango23}. More precisely, \cite{ILO20,Ilango20b} establish $\NP$-hardness of \emph{variants} of $\MCSP$ (such as ``multi-output'' and partial) under \emph{randomized} many-to-one reductions   while the recent work  \cite{Ilango23} shows that $\MCSP$ is $\NP$-hard in the \emph{Randomized Oracle Model} and that $\MCSP^O$ --- the relativized version of $\MCSP$ --- is $\NP$-hard under $\ppoly$ reductions for a random oracle $O$. Finally, $\MCSP$ and, in fact, $\MCSP^B$ can be viewed as special cases of \emph{Natural Properties} \cite{RazborovRudich97,KabanetsCai00} for any oracle $B$.  

The purpose of this paper is to connect $\Avoid$ to $\MCSP$ in a natural setting:
under randomized Turing reductions.
Recently Ghentiyala, Li, and Stephens-Davidowitch \cite{GLS25}
connected it to Arthur--Merlin protocols and showed that $\BPP^{\Avoid}\subseteq\AM\mathrel{\cap}\coAM$
(a corresponding statement holds for the promise versions of these classes).
They also noted that for a modestly large stretch $m=n+\omega(\log n)$
the statement becomes trivial since this oracle $\Avoid_m$
does not bring more power to $\BPP$, namely, $\BPP^{\Avoid_m}=\BPP$.
We ask what is the complexity of $\Avoid$ for other values of $m$,
specifically for the hardest instance of the problem when $m=n+1$?
We show that even in this case
one can replace the $\Avoid$ oracle
with an oracle for $\MCSP$.

While, at face value, $\BPP^{\MCSP}$ and $\AM \cap \coAM$ appear incomparable given our current state of knowledge, an observation of Hirahara and Watanabe on \emph{oracle-independent} reductions to $\MCSP$ \cite{HW16} suggests that $\BPP^{\MCSP}$ may, ``after all'', be contained in $\AM \cap \coAM$. 
For a further discussion on this matter see Section~\ref{sec:oracle independent}.

\vspace{-0.05in}
\subsection{Results}

Our main result shows how to amplify the success probability of solving $\Avoid$ given oracle access to $\MCSP$.

\begin{THEOREM}
\label{THM:main1} 
There exists a randomized algorithm that given $\varepsilon > 0$, a Boolean circuit $C: \bools{n} \to \bools{m}$ (for $m>n$), and oracle access to $\MCSP$, outputs a string $y \in \bools{m}$ such that $y \not \in \Image{C}$, with probability at least $1 - \varepsilon$, in time polynomial in $1/\varepsilon$ and the size of $C$.
\end{THEOREM} 

As an immediate corollary, we obtain an upper bound on the set of all languages (and promise problems) that reduce to $\Avoid$ via randomized Turing reductions.

\begin{COROLARY}\label{cor:cor}
$(\pr)\BPP^{\Avoid} \subseteq (\pr)\BPP^{\MCSP}$.
\end{COROLARY}

\vspace{-0.15in}
\section{Preliminaries}

We begin by formally defining and recalling the relevant problems.

\begin{definition}[{$\protect\Avoid$}]
Given a Boolean circuit $C: \bools{n} \to \bools{m}$ with $m>n$, find an string $y \in \bools{m}$ such that $y \not \in \Image{C}$.
\end{definition}

\begin{definition}[{$\protect\MCSP$, $\protect\MCSP^B$}]
Given the truth table $tt \in \bools{2^n}$ of an $n$-variate Boolean function $f : \bools{n} \to \bools{}$ and a parameter $0 \leq s \leq 2^n$, decide whether $f$ can be computed by a Boolean circuit of size at most $s$. In the relativized version of the problem, $\MCSP^B$, the Boolean circuit is also allowed to use $B$-oracle gates.
\end{definition}

\subsection{Inversion of polynomial-time computable functions}

Allender et al.~\cite{ABKMR06} showed that sufficiently dense sets of strings with sufficiently high time-bounded Kolmogorov complexity (in an appropriate sense)\footnote{%
Namely, one needs a language $L$ containing at least $2^{n}/n^k$ bit strings of each length $n$
such that for every $x \in L, \mathsf{KT}(x) \geq |x|^{1/k}$, where $\mathsf{KT}$ denotes their version of
time-bounded Kolmogorov complexity.%
}
can be used to break (i.e., invert) any polynomial-time computable function. They further observed that such sets can be constructed in
$\PTIME^{\MCSP}$, and hence an $\MCSP$ oracle will suffice for this purpose:

\begin{lemma}[{\protect\cite[Theorem 45 and discussion after that]{ABKMR06}}]\label{lem:Th45}
Let $f_z(x) = f(z, x)$ be a function computable uniformly 
in time polynomial in $\size{x}$. There exists a probabilistic oracle Turing machine $A^\bullet$ and $k \in \N$ such that for any $n$ and $z$:
$$\Pr_{\size{x} = n,\;\tau} \left[ f_z \left( A^{\MCSP}(z, f_z(x), \tau) \right) = f_z(x) \right] \geq 1/n^k,$$
where $x$ is chosen uniformly at random and $\tau$ denotes the randomness of $A$;
this procedure runs in time polynomial in $|x|$ (and $|z|$ and $|f_z(x)|$, but it does not matter due to
the uniformity condition).
\end{lemma}

By observing that evaluating a given Boolean circuit on an input $x$ can be carried uniformly in time polynomial in the size of the circuit, 
and using a standard transformation 
between strong and weak one-way functions we obtain the following:

\begin{lemma}
\label{lem:inv}
There exists a probabilistic oracle Turing machine $M^\bullet$ that given $\varepsilon > 0$ and a Boolean circuit $C$ on $n$ inputs and $m=n+1$ outputs, runs in time $\poly(\size{C}, 1/\varepsilon$) 
and satisfies:
$$\Pr_{\size{x} = n,\; \tau} \left[ C \left( M^{\MCSP}(1/\varepsilon, C, C(x), \tau) \right) = C(x) \right] \geq 1 - \varepsilon,$$
where $x$ is chosen uniformly at random and $\tau$ denotes the randomness of $M$.
\end{lemma}

    Although the result of \cite{ABKMR06} has been used in this form before (see, e.g., \cite{AGMMM18}), for completeness we formally derive it from its original form (Lemma~\ref{lem:Th45})  in Section~\ref{sec:app:miss} of the Appendix. 
    
\begin{remark}
\label{rem:oracle independent}
One can further observe that the $\MCSP$ oracle in Lemmas~\ref{lem:Th45} and~\ref{lem:inv} can be replaced with $\MCSP^B$ for any language $B$ and, in fact, with any exponentially-useful ``natural property'' in the sense of $\cite{RazborovRudich97}$.
This includes $\mathsf{MKTP}$ and other measures.
See further discussion in Section \ref{sec:oracle independent}.
\end{remark}

For notational convenience, we denote the following event of ``successful inversion'':
\begin{definition}
 $I_{C,y,\tau}$ denotes the event that $C \left( M^{\MCSP}(1/\varepsilon, C, y, \tau) \right) = y$. That is, the event that given oracle access to $\MCSP$ and randomness $\tau$, the machine $M$ outputs a pre-image of $y$ under $C$.
\end{definition}

\noindent In particular, note that for any $y \not \in \Image{C}$ and any $\tau$ we have that $\Pr_{\tau} \left[ I_{C,y, \tau} \right] = 0$. Furthermore, given the above notation, we can succinctly rephrase Lemma \ref{lem:inv} as $$\Pr_{\size{x} = n,\, \tau} \left[ I_{C,C(x), \tau} \right] \geq 1 - \varepsilon. $$

Fix a Boolean circuit $C: \bools{n} \to \bools{m}$.
The following lemma relates the probability of inverting a random \emph{input} of $C$ with the probability of inverting a random \emph{element} in the co-domain of $C$, i.e., $\bools{m}$. We include this (very standard) proof for the sake of self-containment:

\begin{lemma}
$\Pr \limits_{\size{y} = m, \tau} \left[ I_{C, y, \tau} \right] \leq \frac{1}{2^{m-n}} \cdot \Pr \limits_{\size{x} = n, \tau} \left[ I_{C,C(x), \tau} \right].$
\end{lemma}
\begin{proof}
\begin{gather*}
\Pr_{x, \tau} \left[ I_{C,C(x), \tau} \right] = 
\sum_{y \in \bools{m}} \Pr_{\tau} \left[ I_{C,y, \tau} \right] \cdot  \Pr_{x} \left[ C(x) = y \right] =
\sum_{y \in \Image{C}} \Pr_{\tau} \left[ I_{C,y, \tau} \right] \cdot  \Pr_{x} \left[ C(x) = y \right] \geq \\ 
\sum_{y \in \Image{C}} \Pr_{\tau} \left[ I_{C,y, \tau} \right] \cdot \frac{1}{2^n} = 
\frac{2^m}{2^n} \cdot \frac{1}{2^m} \cdot \sum_{y \in \bools{m}} \Pr_{\tau} \left[ I_{C,y, \tau} \right] = 
\frac{2^m}{2^n} \cdot \Pr_{y,\tau} \left[ I_{C,y, \tau} \right].\qedhere
\end{gather*}    
\end{proof}

Our next lemma shows that if the inverter from Lemma~\ref{lem:inv} fails to produce a pre-image for a random element $y$ then with high probability this element $y$ has no such pre-image. That is, $y$ lies outside the range of $C$.

\begin{lemma}
\label{lem:succ}
$\Pr \limits_{\size{y} = m, \tau} [y \not \in \Image{C} \mid  \bar{I}_{C, y, \tau} ] \geq  \Pr \limits_{\size{x} = n, \tau} \left[ I_{C,C(x), \tau} \right]$.
\end{lemma}

\begin{proof}
\begin{gather*}
\Pr \limits_{y, \tau} [y \not \in \Image{C} \mid  \bar{I}_{C, y, \tau}] = 
\frac{\Pr \limits_{y, \tau} [y \not \in \Image{C} \wedge  \bar{I}_{C, y, \tau}]}{\Pr \limits_{y, \tau} [\bar{I}_{C, y, \tau}]} = 
\frac{\Pr \limits_{y, \tau} [y \not \in \Image{C}]}{\Pr \limits_{y, \tau} [\bar{I}_{C, y, \tau}]} \geq 
\frac{1 - \frac{1}{2^{m-n}}}{1-\frac{1}{2^{m-n}} \cdot \Pr \limits_{\size{x} = n, \tau} \left[ I_{C,C(x), \tau} \right]} = \\
\frac{2^{m-n}-1}{2^{m-n} - \Pr \limits_{\size{x} = n, \tau} \left[ I_{C,C(x), \tau} \right]} \geq  \Pr \limits_{\size{x} = n, \tau} \left[ I_{C,C(x), \tau} \right].\qedhere
\end{gather*}

\end{proof}

\section{Reducing $\protect\Avoid$ to $\protect\MCSP$: A Proof of the Main Result}

We are now ready to prove our main result.

\begin{proof}[{Proof of Theorem~\protect\ref{THM:main1}.}]
Let $M^\bullet$ be the machine from Lemma~\ref{lem:inv}.
The following procedure solves $\Avoid$ on input $C:\bools{n} \to \bools{m}$ with high probability: 
\begin{enumerate}
    \item Repeat the following $t$ times:
    \begin{enumerate}
    \item Pick $y \in \bools{m}$ uniformly at random.
    \item If  $C \left( M^{\MCSP}(1/\varepsilon,C, y, \tau) \right) \neq y$ then Output $y$; otherwise continue.
    \end{enumerate}
    \item Output $\bot$.
\end{enumerate}
This procedure has two sources of error:
\begin{enumerate}
    \item The procedure outputs an element of $\Image{C}$. By Lemmas~\ref{lem:inv} and~\ref{lem:succ} this probability is at most $\varepsilon t$.
    
    \item The procedure outputs $\bot$. It happens only if all the random picks fall inside $\Image{C}$ (and all the executions of $M^{\MCSP}$ succeed, but it does not really matter).
    This probability is at most $(2^{n-m})^t \leq 1/2^t$ even when $m=n+1$.
\end{enumerate}

The total error is therefore at most $\varepsilon t + 2^{-t}$,
which can be made less than any $\varepsilon'>0$ by choosing $t=n$ and $\varepsilon=\varepsilon'/(2n)$.
\end{proof}

\begin{proof}[Proof of Corollary~\ref{cor:cor}]
Assume that the $\BPP^{\Avoid}$ machine runs in time $p(n)$ for some polynomial $p(n)$.
By choosing $\varepsilon=\frac1{8p(n)}$ in Theorem~\ref{THM:main1},
all queries will be answered correctly with probability at least $\frac78$,
and the error of the machine will remain bounded.
\end{proof}

\section{Oracle-independent reductions and comparison to previous results}
\label{sec:oracle independent}

We remark that in the procedure solving $\Avoid$ described above, the machine $M$ operates correctly not only when given oracle access to $\MCSP$ itself, but also when $\MCSP$ is replaced by $\MCSP^B$ for any \emph{arbitrary} oracle $B$.
This notion of \emph{oracle-independent} reductions to $\MCSP$ was introduced by Hirahara and Watanabe \cite{HW16},
who observed that most known reductions to $\MCSP$ (including those of \cite{ABKMR06}) operate in this manner (see \cite[Section~3]{HW16}). In that work, they also studied the power of such reductions.
In particular, they proved the following:

\begin{nonumtheorem}[{\protect\cite[Theorem 2]{HW16}}]
If a language $L$ is reducible to $\MCSP$ via an oracle-independent randomized reduction
with negligible error that makes at most one query, then $L \in \AM \mathrel{\cap} \coAM$. 
In other words, \[\bigcap_B \BPP^{\MCSP^B[1]} \subseteq \AM \mathrel{\cap} \coAM.\]
\end{nonumtheorem}

They also showed that deterministic, oracle-independent, polynomial-time Turing reductions do not provide additional power, \emph{regardless} of the number of queries:

\begin{nonumtheorem}[{\protect\cite[Theorem 1]{HW16}}]
$\bigcap \limits _B \PTIME^{\MCSP^B} = \PTIME.$
\end{nonumtheorem}

As our procedure for solving $\Avoid$ is oracle-independent 
(see also Remark \ref{rem:oracle independent}) we obtain the following ``oracle-independent'' extension to Corollary \ref{cor:cor}: 

\begin{COROLARY}
$\displaystyle \BPP^{\Avoid} \subseteq \bigcap\limits_B \BPP^{\MCSP^B}$.
\end{COROLARY}

Unfortunately, the above results/characterizations by \cite{HW16} do not apply to our procedure since our reduction is randomized and makes multiple queries. 
Meanwhile, a recent result of Ghentiyala, Li, and Stephens-Davidowitz \cite{GLS25} shows that $\BPP^{\Avoid}\subseteq\AM\mathrel{\cap}\coAM$. 
Taken together, these results provide supporting evidence for the following conjecture, which we put forward:

\begin{conjecture}
$\displaystyle\bigcap\limits_B \BPP^{\MCSP^B} \subseteq \AM \mathrel{\cap} \coAM$.
\end{conjecture}
That is, if a language $L$ is reducible to $\MCSP$ via an oracle-independent randomized reduction then $L \in \AM \mathrel{\cap} \coAM$, regardless of the number of queries.  \\

At the same time, Sdroievski, da Silva, and Vignatti \cite{SSV19} devised a randomized algorithm that, given oracle access to $\MCSP$, solves the Hidden Subgroup Problem ($\HSP$). We note that $\HSP$ is a central problem in the theory of quantum computing. Incidentally, their reduction is also oracle-independent\footnote{Although the formal statement is that $\HSP$ can be solved in $\BPP^\MTKP$, their reduction is actually based on Theorem 45 of \cite{ABKMR06} (Lemma~\ref{lem:inv}). As in our case, this results in an oracle-independent containment in $\BPP^{\MCSP}$.}. Thus, if true, the conjecture would not only imply that our result provides an alternative (and potentially stronger) avenue for placing $\BPP^{\Avoid}$ in $\AM \cap \coAM$, but would also place $\BPP^{\HSP}$ in $\AM \cap \coAM$, which, to the best of our knowledge, is currently unknown. \\

An additional piece of evidence supporting the conjecture (as observed in \cite{HW16}) was provided by the result of Allender and Das \cite{AllenderDas14}, which shows that $\SZK \subseteq \BPP^\MCSP$ via an oracle-independent reduction. This complements the previously known containment $\SZK \subseteq \AM \cap \coAM$ \cite{Fortnow89,Okamoto00,SV03}.
Could the techniques of \cite{GLS25} be extended to prove the conjecture?

\bibliographystyle{alpha}

\newcommand{\etalchar}[1]{$^{#1}$}

\appendix

\section{Proof sketch of Lemma~\protect\ref{lem:inv}}
\label{sec:app:miss}

\begin{proof}[Proof Sketch of Lemma \ref{lem:inv}]
  
    We need essentially to apply a standard construction of converting a weak one-way
    function into a strong one-way function and to take care of the ``uniform'' setting.

    In terms of Lemma~\ref{lem:Th45} we will define the following function $f_z(x)$:
    \[
    f_C(x)=(C(x_1),C(x_2),\ldots,C(x_t)),
    \]
    for $x_i\in\bools{n}$ and for
    large enough $t=O(\frac{n}{\varepsilon})$.
    Then similarly to \cite{Yao82} 
    the following algorithm $M^{\MCSP}$
    succeeds with probability $1-1/\varepsilon$: it repeats 
    an even larger polynomial number of times the following: 
    it takes $j\in\{1,\ldots,t\}$ at random, takes all $x_i$'s at random,
    employs $A^{\MCSP}$ on $(x_1,\ldots,x_t)$ with $x_j$ replaced by $y$,
    and outputs the result if it is correct.

    Now to make our function uniform in $C$ (that is,
    make its time polynomial in $|x|$ only), assume that
    $C$ has at least $\sqrt{|C|}$ inputs.
    If it has fewer inputs, add enough fake inputs (and outputs) to it
    and consider a new circuit $D'$ that computes the identity function on them.
    It won't change the probability of inverting this circuit
    (as $\varepsilon$) is a constant, and the running time of $M^\bullet$ will
    become polynomial in $|C|$ (rather than just in $n$).
\end{proof}

\end{document}